\documentclass[runningheads]{llncs}

\usepackage[T1]{fontenc}
\usepackage{amsmath,latexsym}
\usepackage{amssymb}
\usepackage[utf8]{inputenc}
\usepackage{proof}
\usepackage{rotating}
\usepackage{pifont}
\usepackage{xcolor}
\usepackage[normalem]{ulem}
\usepackage{breqn}
\usepackage{bbding}

\newcommand{\lolli}{\mbox{$-\!\circ$}}
\newcommand{\tensor}{\mathbin{\otimes}}

\newcommand{\calC}{{\cal C}}

\newcommand{\skderiv}{\vdash}
\newcommand{\ljderiv}{\vdash}
\newcommand{\produ}[1]{\langle #1 \rangle}
\newcommand{\subst}[2]{{}^{#1}\!/_{\!#2}}

\newcommand{\FV}{\mathop{FV}}

\newcommand{\down}{\downarrow}
\newcommand{\up}{\uparrow}

\def\bnfas{\mathrel{::=}}
\def\bnfalt{\mid}

\newcommand{\plus}[1]{{#1}^+}
\newcommand{\minus}[1]{{#1}^-}
\newcommand{\posi}[1]{\mathop{pos}(#1)}
\newcommand{\nega}[1]{\mathop{neg}(#1)}
\newcommand{\weak}[1]{\llcorner{#1}\lrcorner}
\newcommand{\aux}[1]{\llcorner{#1}\lrcorner}

\newcounter{beanr}
\newenvironment{enumr}{
\setcounter{beanr}{1}
\begin{list}%
{(\roman{beanr})}{\usecounter{beanr}\labelwidth=6mm\itemindent=0mm\labelsep=2mm\leftmargin=8mm}}{
\end{list}}

\newcommand{\ILL}{\vdash_{\mbox{\upshape \tiny ILL}}}
\newcommand{\R}[1]{\mbox{\upshape #1}}

\newcommand{\fix}[2]{{\bf FIX}\footnote{{\bf #1}: #2}}
\renewcommand{\fix}[2]{}

\begin{document}

\title{Skolemisation for Intuitionistic Linear Logic}      

\author{Alessandro Bruni\inst{1} \and
Eike Ritter\inst{2} \Envelope \and
Carsten Sch\"urmann\inst{1}}

\institute{
   Department of Computer Science \\
   IT University of Copenhagen  \\ 
\texttt{\{brun$|$carsten\}@itu.dk}
  \and
  School of Computer Science \\ 
  University of Birmingham \\ 
 \texttt{E.Ritter@bham.ac.uk}
}
\maketitle              


\fix{ale}{skolem capital or lowercase?}

\begin{abstract}
Focusing is a known technique for reducing the number of proofs while preserving derivability. Skolemisation is another technique designed to improve proof search, which reduces the number of back-tracking steps by representing dependencies on the term level and instantiate witness terms during unification at the axioms or fail with an occurs-check otherwise. Skolemisation for classical logic is well understood, but a practical skolemisation procedure for focused intuitionistic linear logic has been elusive so far. In this paper we present a focused variant of first-order intuitionistic linear logic together with a sound and complete skolemisation procedure.
\end{abstract}
\section{Introduction}

Modern proof search paradigms are built on variants of focused logics first introduced by Andreoli~\cite{andreoli1992logic}.
Focused logics eliminate sources of non-determinism while preserving derivability.
In this paper we consider the focused logic LJF~\cite{Chuck09tcs1}.
By categorising the logical connectives according to the invertibility of its left or right rules, we obtain a so-called polarised logic~\cite{Chuck09tcs1}.
For example, the $\forall$-right rule is invertible, making $\forall$ a negative (or asynchronous) connective, and the $\exists$-left rule is invertible, making $\exists$ a positive (or synchronous) connective.

But even a focused proof system does not eliminate all non-determinism. There is still residual non-determinism in-between focusing steps. It is well known that we can control this non-determinism using different search strategies, such as forcing backward-chaining and forward-chaining using the atom polarity.
Another remaining source of non-determinism comes from the order of quantifier openings, as choosing the wrong order may lead to additional back-tracking.

For example, consider the following judgment in multiplicative linear logic:
\[\forall x. A(x) \lolli B(x), \forall y. \exists u.A(u) \vdash  \exists z. B(z)\]
Variables $u$ introduced by the well-known rules $\exists L^u$ and $\forall R^u$ (and written next to the rule name) are fresh and called Eigen-variables, which we can use to construct witness terms for the universal variables on the left or the existential variables on the right.
Because quantifier rules do not permute freely with other rules, one needs to resolve quantifiers in a particular order, or otherwise risk an exponential blow-up in the proof search.
This fact has already been observed by Shankar~\cite{Shankar91proofsearch} for LJ, who proposed to capture the necessary dependencies using Skolem functions to encode the permutation properties of LJ inference rules, guaranteeing reconstruction of LJ proofs from their skolemised counterparts.

However, na\"ive Skolemisation is unsound in linear logic.
As first noted by Lincoln~\cite{lincoln1995deciding}, the sequent
\[ \forall x. A \tensor B(x) \vdash A \tensor \forall u. B(u) \]
does not admit a derivation in linear logic, but its na\"ive skolemisation does: $A \tensor B(x) \vdash A \tensor B(u())$, where $x$ denotes an existential  and $u()$ a universal variable that must not depend on $x$.
Introducing replication creates a similar problem, where the following sequent does not admit a derivation:
\[ \forall x. !A(x) \vdash !\forall u. A(u) \]
however again its na\"ive skolemisation loses the relative order between quantifier openings and replication, thus admitting a proof: $!A(x) \vdash !A(u())$. 

In this paper we show that the ideas of skolemisation for classical logic\fix{ale}{both logics?} and intuitionistic logic for LJ~\cite{Shankar91proofsearch} carry over quite naturally to focused intuitionistic linear logics (LJF) \cite{Chuck09tcs1}.
We propose a quantifier-free version of LJF that encodes the necessary constraints called skolemised intuitionistic linear logic (SLJF).
Our main contribution is to define a \emph{skolemisation} procedure
from LJF to SLJF that we show to be both sound and complete: any
derivation in LJF is provable in SLJF after skolemisation and, vice
versa, any derivation in SLJF of a skolemised formula allows to
reconstruct a proof of the original formula.
Hence we eliminate back-tracking points introduced by first-order
quantifiers. We do not eliminate any back-tracking points introduced by
propositional formulae.

The paper proceeds as follows: Section~\ref{sec:ill} introduces focused intuitionistic linear logic (LJF), Section~\ref{sec:skill} presents skolemised focused intuitionistic linear logic (SLJF), Section~\ref{sec:skolem} presents a novel skolemisation procedure, Section~\ref{sec:meta} presents  soundness and completeness results, and Section~\ref{sec:conc} presents our conclusion and related work.

\emph{Contributions:} This work is to our knowledge the first work that successfully defines skolemisation for a variant of linear logic. The benefit is that during proof search any back-tracking caused by resolving quantifiers in the wrong order is eliminated and replaced by an admissibility check on the axioms.

\section{Focused Intuitionistic Linear Logic}\label{sec:ill}


We consider the focused and polarised formulation of linear logic LJF~\cite{Chuck09tcs1} that we now present.  The syntactic categories are defined as usual: we write $u, v$ for Eigen-variables and $x, y$ for existential variables that may be instantiated by other terms, finally $N$ for negative formulas and $P$ for positive formulas. We also distinguish between negative and positive atoms, written as $A^-$ and $A^+$. We write $\uparrow$ to embed a positive formula into a negative, and $\downarrow$ for the inverse.  The rest of the connectives should be  self-explanatory.
\[
  \begin{array}{lrcl}
    \mbox{Atom} & A,B &\bnfas&  q(t_1 \dots t_n) \\
    \mbox{Negative formula} &N& \bnfas & A^- \bnfalt P \lolli N 
           \bnfalt \forall x. N \bnfalt \up P \\
    \mbox{Positive formula} &P& \bnfas & A^+ \bnfalt P_1 \tensor P_2 
                                         \bnfalt !N \bnfalt  \exists x. P \bnfalt \down N \\
  \end{array}
\]

We use the standard two-zone notation for judgments with unrestricted
context $\Gamma$ and linear context $\Delta$: we write $\Gamma; \Delta
\vdash N$ for the judgment, where at most one formula $[N] \in \Delta$
or $N=[P]$ can be in focus\fix{ale}{N both cases is confusing}. All
formulas in $\Gamma$ are negative and all other formulas in $\Delta$
are positive. When $[N] \in \Delta$ we say that we focus on the left,
whereas when $N=[P]$ we focus on the right, and we are in an inversion
phase when no formula is in focus. To improve readability, we omit the
leading $\cdot;$ when the unrestricted context is empty. The rules
defining LJF~\cite{Chuck09tcs1} are depicted in Figure~\ref{fig:LJF}.
We comment on a few interesting aspects of this logic. There are two
axiom rules $\minus{\R{ax}}$ and $\plus{\R{ax}}$ where, intuitively,
$\minus{\R{ax}}$ triggers backwards-chaining, and $\plus{\R{ax}}$
forward-chaining~\cite{Pfenning2010oplss}. Hence we can assign
polarities to atoms to select a particular proof search strategy. Once
we focus on a formula, the focus is preserved until a formula with opposite polarity is encountered, in which case the focus is lost or blurred. After blurring, we enter a maximal inversion phase, where all rules without focus are applied bottom-up until no more invertible rules are applicable. The next focusing phase then commences.

Focusing is both sound and complete. i.e. every derivation (written as $\Gamma;\Delta \ILL F$) can be focused and every focused derivation can be embedded into plain linear logic~\cite{Chuck09tcs1}.  In particular, in our own proofs in Section~\ref{sec:meta}, we make use of the soundness of focusing.
\begin{theorem} [Focusing]
  If $\Gamma; \Delta \ILL F$ and $\Gamma'$, $\Delta'$ and $F'$ are the
  result of polarising $\Gamma$, $\Delta$ and $F$ respectively by
  inserting $\uparrow$ and $\downarrow$ appropriately, then $\Gamma';
  \Delta' \vdash F'$ in focused linear logic~\cite{Chuck09tcs1}.
\end{theorem}

We now present three examples of possible derivations of sequents in LJF.
We will use these examples to illustrate key aspects of our proposed skolemisation.
\begin{figure}[t]
  \centering
  \begin{gather*}
    \infer[ax^-]{\Gamma; [A^-] \ljderiv A^-}
    {} \qquad
    \infer[ax^+]{\Gamma; A^+ \ljderiv [A^+]}
    {} \\[1ex]
    \infer[\forall L]
    {\Gamma; \Delta, [\forall\ x.\ N] \ljderiv N'}
    {\Gamma; \Delta, [N\{\subst t x\}] \ljderiv N'} \qquad
    \infer[\forall R^u]
    {\Gamma; \Delta \ljderiv \forall\ u.\ N}
    {\Gamma; \Delta \ljderiv N\{\subst u u\}}\\[1ex]
    \infer[\exists L^u]
    {\Gamma; \Delta, \exists\ u.\ P \ljderiv N'}
    {\Gamma; \Delta, P\{\subst u u\} \ljderiv N'} \qquad
    \infer[\exists R]
    {\Gamma; \Delta \ljderiv [\exists\ x.\ P]}
    {\Gamma; \Delta \ljderiv [P\{\subst t x\}]}\\[1ex]
    \infer[\lolli L]
    {\Gamma; \Delta_1, \Delta_2, [P \lolli N] \ljderiv N'}
    {\Gamma; \Delta_1 \ljderiv [P] & \Gamma; \Delta_2, [N] \ljderiv N'}\qquad
    \infer[\lolli R]
    {\Gamma; \Delta \ljderiv P \lolli N}
    {\Gamma; \Delta, P \ljderiv N}\\[1ex]
    \infer[\tensor L]
    {\Gamma; \Delta, P_1 \tensor P_2 \ljderiv N'}
    {\Gamma; \Delta, P_1, P_2 \ljderiv N'} \qquad
    \infer[\tensor R]
    {\Gamma; \Delta_1, \Delta_2, \ljderiv [P_1 \tensor P_2]}
    {\Gamma; \Delta_1 \ljderiv [P_1] & \Gamma; \Delta_2 \ljderiv [P_2]}\\[1ex]
    \infer[! L]
    {\Gamma; \Delta, !N \ljderiv N'}
    {\Gamma, N; \Delta \ljderiv N'} \qquad
    \infer[! R]
    {\Gamma; \cdot \ljderiv [!N]}
    {\Gamma; \cdot \ljderiv N}
    \qquad
    \infer[copy]
    {\Gamma, N; \Delta \ljderiv N'}
    {\Gamma, N; \Delta, [N] \ljderiv N'} \\[1ex]
    \infer[\mbox{focus}L^*]
    {\Gamma; \Delta, \down N \vdash N'}
    {\Gamma; \Delta, [N] \vdash N'}
    \qquad
    \infer[\mbox{focus}R^*]
    {\Gamma; \Delta \vdash \up P}
    {\Gamma; \Delta \vdash [P]}
    \\[1ex]
    \infer[\mbox{blur} L]
    {\Gamma; \Delta, [\up P] \vdash N'}
    {\Gamma; \Delta, P \vdash N'}
    \qquad
    \infer[\mbox{blur} R]
    {\Gamma; \Delta \vdash [\down N]}
    {\Gamma; \Delta \vdash N}
  \end{gather*}

  \caption{Focused intuitionistic linear logic (LJF)}
  \label{fig:LJF}
\end{figure}



\begin{example} \label{ex:a}
  Consider the motivating formula from the introduction that we would like to derive in LJF, assuming that the term algebra has a term $t_0$.
  \[
    \down(\forall x.  (\down \minus{A(x)}) \lolli \minus{B(x)}), \down(\forall x. \up \exists u.\down \minus{A(u)}) \vdash  \up(\exists x. \down \minus{B(x)})    
  \]
  All formulas are embedded formulas, which means that
  there is a non-deterministic choice to be made, namely on which
  formula to focus next. As this example shows, it is quite important
  to pick the correct formula, otherwise proof search will get stuck
  and back-tracking is required. This observation also holds if we
  determine the instantiation of universal quantifiers on the left
  and existential quantifiers on the right by unification instead of
  choosing suitable terms when applying the $\forall L$ or $\exists R$ rule.

  Focusing on the first assumption before the second will not yield a proof.
  The Eigen-variable that eventually is introduced by the nested existential quantifier inside the second assumption is needed to instantiate the universal quantifier in the first assumption.
  If we start by focusing on the first assumption then none of the subsequent proof states is provable, as the following two proof states  $(\down \minus{A(t_0)}) \lolli \minus{B(t_0)},  \minus{A(t_1)} \vdash  \minus{B(t_0)}$ and  
  $(\down \minus{A(t_0)}) \lolli \minus{B(t_0)},  \minus{A(t_1)} \vdash  \minus{B(t_1)}$ demonstrate.
  Back-tracking becomes inevitable.

  To construct a valid proof we must hence focus on the second assumption before considering the first. The result is a unique and complete proof tree that is depicted in Figure~\ref{fig:a}.
  \begin{figure}
  \[
    \infer*[]
    { \down(\forall x.  (\down \minus{A(x)}) \lolli \minus{B(x)}), \down(\forall x. \up \exists u.\down \minus{A(u)}) \vdash  \up(\exists x. \down \minus{B(x)})}
    { \infer*[]
      {\down(\forall x.  (\down \minus{A(x)}) \lolli \minus{B(x)}), \down \minus{A(u_0)} \vdash  \up(\exists x. \down \minus{B(x)})}
      {\infer*[]
        {\down(\forall x.  (\down \minus{A(x)}) \lolli \minus{B(x)}), \down \minus{A(u_0)} \vdash  \minus{B(u_0)}}
        {\infer[\lolli L]
          {[(\down \minus{A(u_0)}) \lolli \minus{B(u_0)})], \down \minus{A(u_0)} \vdash  \minus{B(u_0)}}
          {
           \infer[\R{blur}R]
           {\down \minus{A(u_0)} \vdash [\down \minus{A(u_0)}]}
           {\infer[\R{focus}L^*]
           {\down \minus{A(u_0)} \vdash \minus{A(u_0)}}
           {\infer[\R{ax}^-]
           {[\minus{A(u_0)}] \vdash \minus{A(u_0)}}
           {}}}
         \quad
           \infer[\R{ax}^-]
           {[\minus{B(u_0)}] \vdash \minus{B(u_0)}}
           {}}
        }
      }
    }
  \]
  \caption{Example~\ref{ex:a}, unique and complete proof}
  \label{fig:a}
  \end{figure}
  \hfill $\Box$
\end{example}

\begin{example}\label{ex:b}
  Consider the sequent
 $ \down (\forall x.  \up (\down\minus{A} \tensor \down \minus{B}(x)))
     \vdash \up (\down\minus{A} \tensor \down \forall u. \minus{B}(u))$.
This sequent is not derivable in LJF: note that
$\forall L$ needs to be above the $\forall R$ rule,
 but this step requires that $\tensor R$ is applied first.  However, to
 apply $\tensor R$, we would need to have applied $\tensor L$ first, which
 requires that $\forall L$ is applied first.  This cyclic dependency cannot be resolved.  \hfill $\Box$
\end{example}

\begin{example}\label{ex:c}
  Consider the sequent
$\down \forall x. \up ! \minus{A} (x) \vdash   \up ! \forall u. \minus{A}(u)$.
This sequent is not derivable in LJF either: note that
the $\forall L$-rule needs to be above the $\forall R$ rule, but this
step requires the $!R$ rule to be applied first. However, to apply the
$!R$ rule we would need to apply the $\forall L$ rule first to ensure
that the linear context is empty when we apply the $!R$ rule. This is another cyclic dependency. \hfill $\Box$
\end{example}

Focusing removes sources of non-determinism from the propositional
layer, but not from quantifier instantiation.  In the next section we
present a quantifier-free skolemised logic, SLJF, where quantifier
dependencies are represented through skolemised terms. This way, proof
search no longer needs to back-track on first-order variables, as the
constraints capture all dependencies. Instead, unification at the
axioms will check if the proof is admissible.

\section{Skolemised Focused Intuitionistic Linear Logic} \label{sec:skill}

We begin now with the definition of a skolemised, focused, and polarised intuitionistic linear logic (SLJF), with the following syntactic categories:
\[
  \begin{array}{lrcl}
    \mbox{Atom} & A,B &\bnfas&  q(t_1 \dots t_n) \\
    \mbox{Negative formula} & N &\bnfas& \minus{A}_\Phi
                                         \bnfalt 
                                         P \lolli N  \bnfalt \up P \\
    \mbox{Positive formula} & P &\bnfas& \plus{A}_\Phi
                                         \bnfalt P \tensor P
                                         \bnfalt !_{(a;\Phi; \sigma)} N
                                         \bnfalt \down N \\[1ex]
    \mbox{Variable} & v & \bnfas & x \bnfalt u \bnfalt a \\
    \mbox{Term} & t & \bnfas & v \bnfalt  f(t) \bnfalt (t,
                               \ldots, t)\\
    \mbox{Variable context} & \Phi &\bnfas  & \cdot \bnfalt \Phi, v\\
   \mbox{Modal context} & \Gamma &\bnfas & \cdot \bnfalt
                                                  \Gamma,  (a;\Phi; \sigma):N \\
    \mbox{Linear context} & \Delta &\bnfas & \cdot \bnfalt \Delta, P \\[1ex]
  \mbox{Parallel substitution} & \sigma &\bnfas & \cdot \bnfalt
                                                  \sigma, t/x \bnfalt
                                                   \sigma, u(t)/u \bnfalt
                                                  \sigma, t/a 
  \end{array}
\]
Following the definition of LJF, we distinguish between positive and negative formulas and atoms. Backward and forward-chaining strategies are supported in SLJF, as well.

SLJF does not define any quantifiers as they are removed by skolemisation (see Section~\ref{sec:skolem}). Yet, dependencies need to be captured in some way. Quantifier rules for $\forall R^u$ and $\exists L^u$ introduce Eigen-variables written as $u$. Quantifier rules for $\forall L$ and $\exists R$ introduce existential variables, which we denote with $x$. And finally other rules, such as $\tensor R$, $\lolli L$, and $! R$ are annotated with \emph{special variables} $a$ capturing the dependencies between rules that do not freely commute. These special variables are crucial during unification at the axiom level to check that the current derivation is admissible. 



The semantics of the bang connective $!$ in SLJF is more involved than
in LJF because we have to keep track of the variables capturing
dependencies and form closures: One way to define the judgmental
reconstruction of the exponential fragment of SLJF is to introduce a
validity judgment $(a;\Phi;\sigma) : N$, read as $N$ is valid in world
$(a;\Phi;\sigma)$, which leads to a generalised, modal $\Gamma$ that
no longer simply contains negative formulas $N$, but also closures of
additional judgmental information.  The special variable $a$ is the
``name'' of the world in which $N\sigma$ is valid, where all
possible dependencies are summarised by $\Phi$. $\Phi$ consists of
variables, where we assume tacit variable renaming to ensure that no
variable name occurs twice. We write $\aux{\Phi}$ for all existential
and special variables declared in $\Phi$.  In contrast to LJF, atomic
propositions $A^-_\Phi$ and $A^+_\Phi$ are indexed by $\Phi$ capturing
all potential dependencies, which we will inspect in detail in
Definition~\ref{def:adm} where we define \emph{admissibility}, the
central definition of this paper, resolving the non-determinism related
to the order in which  quantifier rules are applied.
The linear context remains unchanged.


Terms $t$ are constructed from variables (existential, universal, and
special) and function symbols $f$ that are declared in a global
signature $\Sigma \bnfas \cdot \bnfalt \Sigma, f$.  Well-built terms
are characterised by the judgment $\Phi \vdash t$. Substitutions
constructed by unification and communicated through proof search
capture the constraints on the order of application of proof rules,
which guarantee that a proof in SLJF gives rise to a proof in LJF.
Their definition is straightforward, and the typing rules for
substitutions are depicted in Figure~\ref{fig:subst}. For a
substitution $\sigma$ such that $\sigma \colon \Phi \rightarrow
\Phi'$, we define the domain of $\sigma$ to be $\Phi$ and the
co-domain of $\sigma$ to be $\Phi'$. For any context $\Phi$ and
substitution $\sigma$ with co-domain $\Psi$ we write $\sigma_{\uparrow
  \Phi}$ for the substitution $\sigma$ restricted to $\Phi \cap \Psi$,
i.e.~$v\sigma_{\uparrow \Phi}$ is defined iff $v \in \Phi \cap \Psi$,
and $v\sigma_{\uparrow \Phi} = v\sigma$ in this case. We write $\sigma
\setminus \Phi$ for the substitution $\sigma$ restricted to $\Psi
\setminus \Phi$, i.e.~$v\sigma\setminus\Phi$ is defined iff $v \in
\Psi \setminus \Phi$, and $v\sigma\setminus \Phi = v\sigma$ in this
case.
For any substitution $\sigma$ we define the substitution $\sigma^n$ by
induction over $n$ to be $\sigma^1 = \sigma$, and $v\sigma^{n+1} = (v\sigma^n)\sigma$.

\begin{figure}[t]
  \[
    \begin{array}{cc}
  \infer[s/\cdot]{\cdot \colon \Phi \rightarrow \cdot}{}  &
\infer[\R{s/existential}]{\sigma, t/x\colon \Phi \rightarrow    \Phi', x}
       {\Phi \vdash t & \sigma \colon \Phi
         \rightarrow \Phi'} \\[2ex]
\infer[\R{s/Eigen}]{\sigma, u(\vec{t})/u \colon \Phi \rightarrow \Phi',
  u}{\Phi \vdash \vec{t} & \sigma \colon \Phi \rightarrow \Phi'} &
\infer[\R{s/special}]{\sigma, t/a\colon \Phi \rightarrow    \Phi', a}
       {\Phi \vdash t & \sigma \colon \Phi
                        \rightarrow \Phi'}
    \end{array}
  \]
   \caption{Typing rules for substitutions}
   \label{fig:subst}
 \end{figure}

 \begin{definition}[Free Variables]
    We define the free variables of a skolemised formula $K$, written $\mathop{FV}(K)$ by
    induction over the structure of formulae by
    \[\begin{array}{rcl}
     
      \FV(\minus{A}_{\Phi}) &=& \FV(\plus{A}_{\Phi}) = \Phi \\      
      \FV(P_1 \tensor P_2) &=& \FV(P_1) \cup \FV(P_2) \\
      \FV(P \lolli N) &=& \FV(P) \cup \FV(N) \\
      \FV(!_{(a;\Phi; \sigma)}N) &=& \Phi \\
      
      \end{array}\]
\end{definition}
Now we turn to the definition of admissibility, which checks whether
the constraints on the order of $\forall L$-and $\exists R$-rules
(which instantiate quantifiers) and application of non-invertible propositional rules can be satisfied
when re-constructing a LJF-derivation from an SLJF-derivation.
\begin{definition}[Admissibility] \label{def:adm} We say $\sigma$ is
  \emph{admissible} for $\Phi$ if firstly for all existential and special variables $v$ and for all
  $n$, $v$ does not occur in $v\sigma^n$, and secondly for all special variables $a_L$ and
  $a_R$ respectively and for all $n$, if $x\sigma^n$ contains a variable $a_L$ or $a_R$
  for any $x$ in the co-domain of $\sigma$, then the
  variable $a_R$ or $a_L$ respectively does not occur in $\Phi$.
\end{definition}
The first condition in the definition of admissibility ensures that 
there are no cycles in the dependencies of  $\forall L$-and $\exists R$-rules and
non-invertible propositional rules. The second condition ensures
that for each rule with two premises any Eigen-variable which is
introduced in one branch is not used in the other branch. Examples of
how this definition captures dependency constraints are given below.

Next, we define derivability in SLJF. The derivability judgment uses
a substitution which captures the dependencies between  $\forall L$-and $\exists R$-rules and non-invertible propositional rules.
\begin{definition}[Proof Theory] Let $\Phi$ be a context of variables, $\Gamma$ the modal context (which refined the notion of unrestricted context from earlier in this paper), $\Delta$ the linear context, $P$ a positive and $N$ a negative formula, and $\sigma$ a substitution. We define two mutually dependent judgments $\Gamma;\Delta \skderiv N;\sigma$ and $\Gamma;\Delta \skderiv [P];\sigma$ to characterise derivability in SLJF. The rules defining these judgments are depicted in Figure~\ref{fig:sklogic}.
\end{definition}
The $!R$-rule introduces additional substitutions which capture the
dependency of the $!R$-rule on the  $\forall L$-and $\exists R$-rules
which instantiate the free variables in the judgment. An example of
this rule is given below. The copy-rule performs a renaming of all the
bound variables in $N$. 


\begin{figure}[ht]
\centering
\begin{gather*}
    \infer[\minus{ax}]{\Gamma; [\minus{A}_{\Phi_1}] \skderiv \minus{B}_{\Phi_2}; \sigma}
    {\minus{A}\sigma = \minus{B}\sigma &  \sigma \mbox{ admissible for
        $\Phi_1, \Phi_2$}} \\[1ex]
     \infer[ax^+]{\Gamma; \plus{A}_{\Phi_1} \skderiv [\plus{B}_{\Phi_2}]; \sigma}
     {\plus{A}\sigma = \plus{B}\sigma & \sigma \mbox{ admissible for
        $\Phi_1, \Phi_2$}} \\[1ex]
     \infer[\lolli L]
     {\Gamma; \Delta_1, \Delta_2, [P \lolli N] \skderiv N';
       \sigma}
     {\Gamma; \Delta_1 \skderiv [P]; \sigma &
      \Gamma; \Delta_2, [N] \skderiv N'; \sigma} \\[1ex]
     \infer[\tensor R]
     {\Gamma; \Delta_1, \Delta_2 \skderiv [P_1 \tensor P_2];\sigma }
     {\Gamma; \Delta_1\skderiv [P_1]; \sigma &
      \Gamma; \Delta_2 \skderiv [P_2];\sigma} \\[1ex]
     \infer[\lolli R]
     {\Gamma; \Delta \skderiv P \lolli N; \sigma}
     {\Gamma; \Delta, P \skderiv N; \sigma} \qquad
     \infer[\tensor L]
     {\Gamma; \Delta, P_1 \tensor P_2 \skderiv N; \sigma}
     {\Gamma; \Delta, P_1, P_2 \skderiv N;    \sigma} 
 \\[1ex]
     \infer[!L]
     {\Gamma; \Delta, !_{(a;\Phi; \sigma')} N \skderiv N'; \sigma}
     {\Gamma, (a;\Phi; \sigma')\colon N; \Delta \skderiv N'; \sigma} \\[1ex]
     \infer[!R]
     {(a_i;\Phi_i;\sigma_i)\colon N_i; \cdot \skderiv [!_{(a,\Phi; \sigma')} N]; \sigma}
     {(a_i;\Phi_i;\sigma_i)\colon N_i; \cdot \skderiv N;
       \sigma, \sigma', (\Phi_1, \ldots, \Phi_n)/a, (\Phi)/a_1,
       \ldots, (\Phi)/a_n}\\[1ex]
    \infer[copy]{\Gamma, (a;\Phi; \sigma')\colon N; \Delta \skderiv N';
      \sigma}
    {\Gamma,(a;\Phi; \sigma')\colon N; \Delta,
      [N\{\vec{v}'/\vec{v}\}]  \skderiv N';
      \sigma, \sigma'\{\vec{v}'/\vec{v}\}\; \mbox{where $\vec{v} = FV(N)\setminus\Phi$}} \\[1ex]
    \infer[\mbox{focus}L^*]
    {\Gamma; \Delta, \down N \skderiv N'; \sigma}
    {\Gamma; \Delta, [N] \skderiv N'; \sigma}
    \qquad
    \infer[\mbox{focus}R^*]
    {\Gamma; \Delta \skderiv \up P; \sigma}
    {\Gamma; \Delta \skderiv [P]; \sigma}
    \\[1ex]
    \infer[\mbox{blur} L]
    {\Gamma; \Delta, [\up P] \skderiv N'; \sigma}
    {\Gamma; \Delta, P \skderiv N'; \sigma}
    \qquad
    \infer[\mbox{blur} R]
    {\Gamma; \Delta \skderiv [\down N];\sigma}
    {\Gamma; \Delta \skderiv N; \sigma} 
  \end{gather*}
  \caption{Skolemised intuitionistic linear logic}
  \label{fig:sklogic}
\end{figure}


\begin{example} \label{ex:e}
  We give a derivation of the translation of the judgment of Example~\ref{ex:a}
  in skolemised intuitionistic
  linear logic. We omit the modal context 
$\Gamma = \cdot$. Furthermore, let the goal
  of proof search be the following judgment:
   \[\begin{array}{l}
   \cdot;
    \up ( \down A(x_1)^-_{(x_1, a_L)}) \lolli B(x_1)^-_{( x_1, a_R)},
    \up A(u)^-_{(x_2,u)}
      \vdash  B(x_3)^-_{(x_3)}; \sigma
    \end{array}
  \]
  where $\sigma$ must contain the substitution $u(x_2)/u$, which
  arises from skolemisation.
\end{example}
We observe that only focusing rules are applicable. Focusing on $A$  will not succeed, since $A$ was assumed to be a
negative connective, so we focus on the right. Recall, that we will
not be able to remove the non-determinism introduced on the
propositional level. We obtain the derivation in
Figure~\ref{fig:uniqueSKProof}, where $\sigma = \cdot, u/x_1, x_1/x_3,
u(x_2)/u$. This derivation holds because $\sigma$ is admissible for $x_1, a_L,
x_2$ and $x_1, a_R, x_3$. The constraint that the variable $x_2$ can
be instantiated only after the $\forall R$-rule for $u$ has been
applied is captured by the substitution $u(x_2)/u$.
\begin{figure}
\begin{center}
$\begin{array}{l}
    \infer[]
    {\down ( \down A(x_1)^-_{(x_1, a_L)}) \lolli B(x_1)^-_{( x_1, a_R)},
    \down A(u)^-_{(x_2,u)}
      \vdash  B(x_3)^-_{(x_3)};\sigma}
    {
    \infer[]
    {[( \down A(x_1)^-_{(x_1, a_L)}) \lolli B(x_1)^-_{( x_1, a_R)}],
    \down A(u)^-_{(x_2,u)}
      \vdash  B(x_3)^-_{(x_3)};\sigma}
    {
    \infer[]
    {    \down A(u)^-_{(x_2,u)}
      \vdash   [\down A(x_1)^-_{(x_1, a_L)}];\sigma}
    {
        [A(u)^-_{(x_2,u)}]
      \vdash  A(x_1)_{(x_1, a_L)};\sigma
    }
    \quad
    [B(x_1)^-_{( x_1, a_R)}]
      \vdash  B(x_3)^-_{(x_3)};\sigma
    }
    }
  \end{array}$
\end{center}
\caption{Example~\ref{ex:e}, unique complete proof}
\label{fig:uniqueSKProof}
\end{figure}

\begin{example} \label{ex:admissibility}
  Next, consider the sequent $\down (\forall x. \up(\down \minus{A} \tensor
  \down \minus{B}{x})) \vdash \up (\down \minus{A} \tensor \down \forall
  u.\minus{B}(u))$ from Example~\ref{ex:b}.
  To learn if this sequent is provable, we translate it into
$\down\minus{A}_x \tensor \down \minus{B(x)}_x \skderiv \up (\down
A_{a_L} \tensor \down \minus{B(u)}_{a_R;u})$. The only possible proof
yields an axiom derivation $[\minus{B}_{x}]
\skderiv \minus{B}_{a_R;u}; \cdot, u(a_R)/x$ , which is not
valid, as $\cdot, u/x, u(a_R)/u$ is not admissible for $x,a_L$. More
precisely, the second condition of admissibility is violated. 
\hfill $\Box$
\end{example}

\begin{example} \label{ex:admPling}
  Now, consider the sequent $\down \forall x. \up ! \minus{A} (x) \vdash
  \up ! \forall u. \minus{A}(u)$ from Example~\ref{ex:c}. The
  skolemised sequent is
  $!_{(a;x;\cdot)}\minus{A_{a,x}} \skderiv \up !_{(b;u;u(b)/u)}
  \minus{A}(u)_{u,b})$. The only possible derivation produces the
  substitution $\cdot, u/x, x/b, u(b)/u$, which is not admissible for
  $\cdot, x, u, b, a$. More precisely, the first condition of
  admissibility is violated for the variable $b$. This expresses the
  fact that in any possible LJF-derivation the instantiation of $x$ has
  to happen before the $!R$-rule and the $!R$ -rule has to be applied
  before the instantiation of $x$, which is impossible.
\end{example}

\section{Skolemisation} \label{sec:skolem}

To skolemise first-order formulas in classical logic, we usually compute prenex normal forms of all formulas that occur in a sequent, where we replace all quantifiers that bind ``existential'' variables by Skolem constants. This idea can also be extended to intuitionistic logic~\cite{Shankar91proofsearch}. This paper is to our knowledge the first  to demonstrate that skolemisation can also be defined for focused, polarised, intuitionistic, first-order linear logic, as well.   In this section, we show how.

Skolemisation transforms an LJF formula $F$ (positive or negative) closed under $\Phi$ into an SLJF formula $K$ and a substitution, which
collects all variables introduced during skolemisation. Formally, we define two mutual judgments:  $sk_L(\Phi,F) = (K; \sigma)$ and $sk_R(\Phi,F) = (K; \sigma)$. $K$ is agnostic to polarity information, hence we prepend appropriate $\uparrow$ and $\downarrow$ connectives to convert $K$ to the appropriate polarity by the conversion operations $\posi{\cdot}$ and $\nega{\cdot}$, depicted in Figure~\ref{fig:polad}. Alternatively, we could have chosen to distinguish positive and negative $K$s syntactically, but this would have unnecessarily cluttered the presentation and left unnecessary backtrack points because of spurious $\uparrow \downarrow$  and $\downarrow \uparrow$ conversions.
 

\begin{figure}[t]
  \begin{center}
    \begin{minipage}{2in}
      \[
        \begin{array}{rcl}
          \posi{A^-} & = &  \down A^-{} \\
          \posi{P \lolli N} & = &  \down (P \lolli N) \\
          \posi{\up P} & = & P \\
          \posi{A^+} & = &  A^+ \\
          \posi{\down N} & = & \down N \\
          \posi{!_{(a;\Phi; \sigma)}N} & = & !_{(a;\Phi; \sigma)}N
        \end{array}
      \]
    \end{minipage}
    \begin{minipage}{2in}
      \[
        \begin{array}{rcl}
          \nega{A^-} & = &  A^- \\
          \nega{P \lolli N} & = &  P \lolli N \\
         \nega{\up P} & = & \up P \\
          \nega{A^+} & = &  \up A^+ \\
         \nega{P_1 \tensor P_2} & = & \up (P_1 \tensor P_2) \\
          \nega{\down N} & = & N \\
          \nega{!_{(a;\Phi; \sigma)}N} & = & \up !_{(a;\Phi; \sigma)}N
        \end{array}
      \]
    \end{minipage}
  \end{center}
  \caption{Polarity adjustments}
  \label{fig:polad}
\end{figure}

We return to the definition of skolemisation, depicted in Figure~\ref{fig:sk}.
The main idea behind skolemisation is to record dependencies of
quantifier rules as explicit substitutions. More precisely, if an
Eigen-variable $u$ depends on an existential variable $x$, a substitution
$u(x)/u$ is added during skolemisation. We do not extend the scope of an
Eigen-variable beyond the !-operator as we have to distinguish between
an Eigen-variable for which a new instance must be created by the
\texttt{copy}-rule and one where the same instance may be retained.

Explicit substitutions model constraints on the order of
quantifiers. The satisfiability of the constraints is checked
during unification at the leaves via the admissibility condition (see Definition~\ref{def:adm}) which the
substitution has to satisfy. Potential back-track points are marked by special
variables $a$, which are associated with the
$!$ connective.  These annotations need to store enough information so that the set
of constraints can be appropriately updated when copying a formula
from the modal context into the linear context.


In our representation,  any proof of the skolemised formula in SLJF
captures an equivalence class of proofs under different quantifier
orderings in LJF. Only those derivations where substitutions are
admissible, i.e.~do not give rise to cycles like $u(x)/x$ or introduce undue dependencies between the left and right branches of a $\tensor$ or $\lolli$, imply the existence of a proof in LJF.

The judgments can be easily extended to the case of contexts $\Gamma$ and $\Delta$ for which we write $ sk_L(\Phi;\Gamma)$ and  $ sk_L(\Phi;\Delta)$. Note that tacit variable renaming is in order, to make sure that no spurious cycles are accidentally introduced in the partial order defined by the constraints.

\newcommand{\skl}[2]{\mbox{sk}_L({#1};{#2})}
\newcommand{\skr}[2]{\mbox{sk}_R({#1};{#2})}

\begin{figure}[t]
  \resizebox{\textwidth}{!}{
  \begin{math}\begin{aligned}
    &\skl{\Phi}{A} = \posi{A_\Phi}; \cdot && \skr{\Phi}{A} = \nega{A_\Phi}; \cdot \\[1ex]
    &\skl{\Phi}{\forall x. F} = K; \sigma && \skr{\Phi}{\exists x. F} = K; \sigma\\
    &\qquad\mathsf{where}\ \skl{x,\Phi}{F} = K; \sigma && \qquad\mathsf{where}\ \skr{x, \Phi}{F} = K; \sigma \\[1ex]
    &\skl{\Phi}{\exists u.F} = K; \sigma, u(\weak{\Phi})/u && \skr{\Phi}{\forall u. F} = K; \sigma, u(\weak{\Phi})/u \\
    &\qquad\mathsf{where}\ \skl{u,\Phi}{F} = K; \sigma && \qquad\mathsf{where}\ \skr{u,\Phi}{F} = K; \sigma \\[1ex]
    &\skl{\Phi}{F_1 \tensor F_2} = \posi{K_1} \tensor \posi{K_2}; \sigma_1, \sigma_2 && \skr{\Phi}{F_1 \tensor F_2} = \posi{K_1} \tensor \posi{K_2}; \sigma_1, \sigma_2\\
    &\qquad\mathsf{where}\ \skl{\Phi}{F_1} = K_1; \sigma_1 && \qquad\mathsf{where}\ \skr{\Phi, a_L}{F_1} = K_1; \sigma_1 \\
    &\hspace{5em}\skl{\Phi}{F_2} = K_2;\sigma_2 && \hspace{5em}\skr{\Phi, a_R}{F_2} = K_2; \sigma_2\\[1ex]
    &\skl{\Phi}{F_1 \lolli F_2} = \posi{K_1} \lolli \nega{K_2}; \sigma_1, \sigma_2 && \skr{\Phi}{F_1 \lolli F_2} = \posi{K_1} \lolli \nega{K_2}; \sigma_1, \sigma_2\\
    &\qquad\mathsf{where}\ {\skr{\Phi, a_L}{F_1} = K_1; \sigma_1}&& \qquad\mathsf{where}\ {\skl{\Phi}{F_1} = K_1;\sigma_1}\\
    &\hspace{5em}\skl{\Phi, a_R}{F_2} = K_2; \sigma_2 && \hspace{5em}\skr{\Phi}{F_2} = K_2; \sigma_2\\[1ex]
    &\skl{\Phi}{!F} = !_{(a;\Phi; \sigma \setminus \Phi)} \nega{K}; \sigma_{\uparrow \Phi} &&  \skr{\Phi}{!F} = !_{(a,\Phi; \sigma \setminus \Phi)}\nega{K};\sigma_{\uparrow \Phi} \\
    &\qquad\mathsf{where}\ {\skl{\Phi, a}{F} = K; \sigma} && \qquad\mathsf{where}\ {\skr{\Phi,a}{F} = K; \sigma} \\[1ex]
    &\skl{\Phi}{\downarrow F} = \nega{K};\sigma && \skr{\Phi}{\downarrow F} = \nega{K}; \sigma \\
    &\qquad\mathsf{where}\ {\skl{\Phi}{F} = K; \sigma} && \qquad\mathsf{where}\ {\skr{\Phi}{F} = K; \sigma}\\[1ex]
    &\skl{\Phi}{\uparrow F} = \posi{K}; \sigma && \skr{\Phi}{\uparrow F} = \posi{K}; \sigma \\
    &\qquad\mathsf{where}\ {\skl{\Phi}{F} = K; \sigma} && \qquad\mathsf{where}\ {\skr{\Phi}{F} = K; \sigma}
  \end{aligned}\end{math}}
  \caption{Skolemisation}
  \label{fig:sk}
\end{figure}

\begin{example} \label{ex:d} We return to Example~\ref{ex:a} and simply present the skolemisation of the three formulas that define the judgment:
    \[
    \down(\forall x.  (\down A(x)) \lolli B(x)), \down(\forall x. \up \exists u.\down A(u)) \vdash  \up(\exists x. \down B(x))    
  \]
  First, we skolemise each of the formulas individually.
  \begin{eqnarray*}
    \skl{\cdot}{\down(\forall x.  (\down A(x)) \lolli B(x))} 
    &=& ( \down A(x)_{( x, a_L)}) \lolli
    B(x)_{( x, a_R)}; \cdot\\
    \skl{\cdot}{\down(\forall x. \up \exists u.\down A(u))} &=&
      A(u)_{(x,u)} ; u(x)/u \\
    \skr{\cdot}{\up(\exists x. \down B(x))} &=&  B(x)_{(x)}; \cdot 
  \end{eqnarray*}
  Second, we assemble the results into a judgment in SLJF, which then looks as follows. To  this end, we $\alpha$-convert the variables, 
  \[\begin{array}{l}
    ( \down A(x_1)_{(x_1, a_L)}) \lolli B(x_1)_{( x_1, a_R)},
    A(u)_{(x_2,u)}
      \vdash  B(x_3)_{(x_3)}; u(x_2)/u
    \end{array}
  \]
  The attentive reader might have noticed that we already gave a proof of this judgment in the previous section in Example~\ref{ex:a}, after turning the first two formulas positive, because they constitute the linear context.
\end{example}

\section{Meta Theory} \label{sec:meta}

We begin now with the presentation of the soundness result (see
Section~\ref{sec:sound}) and the completeness result (see
Section~\ref{sec:comp}). Together they imply that skolemisation
preserves provability.  These theorems also imply that proof search in SLJF will be more efficient than in LJF since it avoids quantifier level back-tracking.  Proof search in skolemised form will not miss any solutions.

\subsection{Soundness} \label{sec:sound}
For the soundness direction, we show that any valid derivation in LJF
can be translated into a valid derivation in SLJF after
skolemisation.

\begin{lemma} [Weakening]
  \label{lemma:weakening}
  \begin{enumr}
    \item
  Assume 
$\Gamma; \Delta
\skderiv K;\sigma$
Then also
$\Gamma, (a;\Phi; \sigma') \colon N; \Delta
\skderiv K;\sigma$.
\item
Assume $\Gamma; \Delta
\skderiv [K];\sigma$
Then also
$\Gamma, (a;\Phi; \sigma') \colon N; \Delta
\skderiv [K];\sigma$.
\item
Assume $\Gamma; \Delta, [K']
\skderiv K;\sigma$
Then also
$\Gamma, (a;\Phi; \sigma') \colon N; \Delta, [K']
\skderiv K;\sigma$.
\end{enumr}

\end{lemma}
\begin{proof}
The proof is a simple induction over derivation in all three cases.
\end{proof}

Next, we prove three admissibility properties for $\tensor$R, $\lolli$L, and $copy$, respectively, that we will invoke from within the proof of the soundness theorem. In the interest of space, we provide a proof only for the first of the three lemmas.

\begin{lemma}[Admissibility of $\tensor$R]
\label{lemma:tensor}
Assume $\Gamma; \Delta_1 \skderiv \nega{K_1};
  \sigma$ and $\Gamma; \Delta_2
  \skderiv \nega{K_2};\sigma$ with proofs of
  height at most $n$ such that the first application of the focus-rule 
  is the focus R-rule. Then also
  $\Gamma; \Delta_1 ,\Delta_2 \skderiv
  \nega{\posi{K_1} \tensor \posi{K_2}}; \sigma$. 
\end{lemma}
\begin{proof}
We prove this property by induction over $n$. 
  There are several cases. Firstly, assume that there is any
  positive formula in $\Delta_1$ or $\Delta_2$ which is not an
  atom. Again, there are several cases. We start by 
  assuming $\Delta_1 = K_1' \tensor K_2', \Delta_1'$ and the derivation
  is
\[
\infer{\Gamma; K_1' \tensor K_2', \Delta_1' \skderiv
  \nega{K_1};  \sigma}
   {\Gamma; K_1',K_2', \Delta_1' \skderiv
  \nega{K_1};  \sigma}
\]
 Hence by
induction hypothesis we have 
$\Gamma; K_1', K_2', \Delta_1', \Delta_2 \skderiv
  \nega{\posi{K_1} \tensor \posi{K_2}}; \sigma
$
and hence also 
$\Gamma; K_1' \tensor K_2', \Delta_1', \Delta_2 \skderiv
  \nega{\posi{K_1} \tensor \posi{K_2}};  \sigma \; .
$
Now assume that $\Delta_1 = !_{(a;\Phi; \sigma')}N, \Delta_1'$ and the
derivation is
\[
\infer{\Gamma; !_{(a;\Phi; \sigma')}N, \Delta_1' \skderiv
  \nega{K_1};  \sigma}
   {\Gamma, (a;\Phi; \sigma') \colon N; \Delta_1' \skderiv
  \nega{K_1};  \sigma}
\]
By Lemma~\ref{lemma:weakening}, we also have $\Gamma, (a;\Phi; \sigma') \colon N; \Delta_2
  \skderiv \nega{K_2};\sigma$.
By induction hypothesis we have
$\Gamma, (a;\Phi; \sigma') \colon N; \Delta_1', \Delta_2 \skderiv
  \nega{\posi{K_1} \tensor \posi{K_2}};  \sigma \; 
$
and hence also
$\Gamma; !_{(a;\Phi; \sigma')}N,  \Delta_1', \Delta_2 \skderiv
  \nega{\posi{K_1} \tensor \posi{K_2}};  \sigma \; . 
$

Secondly, assume that $K_1= N_1$ , where $N_1$ is a negative formula
and $K_2= P_2$, where $P_2$ is a positive formula.
By assumption there is a derivation
\[
\infer{\Gamma; \Delta_2 \skderiv \up P_2;
  \sigma}
{\Gamma, \Delta_2 \skderiv  [P_2]
  \sigma}
\]
There is also a derivation
\[
\infer{\Gamma; \Delta_1 \skderiv [\down N_1];
         \sigma}
       {\Gamma; \Delta_1 \skderiv N_1;
         \sigma} 
\]
Hence we also have the following derivation:
\[
\infer{\Gamma; \Delta_1, \Delta_2 \skderiv \up (\down N_1
  \tensor P_2); \sigma}
      {\infer{\Gamma; \Delta_1, \Delta_2 \skderiv  [\down N_1
  \tensor P_2]; \sigma}
             {\infer{\Gamma; \Delta_1 \skderiv  [\down N_1];\sigma}
                     {\Gamma; \Delta_1 \skderiv  N_1;
               \sigma}\;\;
              \Gamma; \Delta_2 \skderiv [P_2];
               \sigma}}
         \]
By assumption we obtain $\Gamma; \Delta_1, \Delta_2 \skderiv \up (\down N_1
\tensor P_2); \sigma$.
All other cases of $K_1$ and $K_2$ being positive or negative are similar.
\end{proof}

\begin{lemma}[Admissibility of $\lolli$L]
\label{lemma:lolli}
Assume $$\Gamma; \Delta_1 \skderiv \nega{K_1};
  \sigma \;\; \mbox{and} \;\;\Gamma; \Delta_2, \posi{K_2}
  \skderiv K;\sigma$$ with proofs of
  height at most $n$ such that the first application of the focus-rule 
  is the focus L-rule for $K_1$. and the focus R-rule for $K_2$. Then also
  \[\Gamma; \Delta_1 ,\Delta_2  , \nega{\posi{K_1}
    \lolli\posi{K_2}}\skderiv K
; \sigma\]
\end{lemma}
\begin{proof}
  Similar to the proof of Lemma~\ref{lemma:tensor}.
\hfill $\Box$
\end{proof}
\begin{lemma}[Admissibility of $copy$]
\label{lemma:copy}
Assume $$\Gamma, (a;\Phi; \sigma')\colon N;
\posi{N\{\vec{v}'/\vec{v}\}}, \Delta \skderiv \nega{K}; 
  \sigma, \sigma'\{\vec{v}'/\vec{v}\}$$ with a proof of
  height at most $n$ such that the first application of the focus-rule 
  is the focus L-rule applied to $\posi{N\{\vec{v}'/\vec{v}\}}$. Then also
  $\Gamma, (a;\Phi; \sigma')\colon N; \Delta \skderiv
  \nega{K} ; \sigma$.
\end{lemma}
\begin{proof}
  Similar to the proof of Lemma~\ref{lemma:tensor}.
\hfill $\Box$
\end{proof}
\begin{theorem}[Soundness]
Let $\Phi$ be a context which contains all the free
  variables of $\Gamma$, $\Delta$ and $F$. Let $\sigma\colon \Phi \rightarrow \Phi $ be a
  substitution. Assume  $\Gamma\sigma_{\uparrow \weak{\Phi}};
  \Delta\sigma_{\uparrow \weak{\Phi}} \vdash F\sigma_{\uparrow \weak{\Phi}}$ in focused intuitionistic linear logic. 
Let
  $sk_L(\Phi; \Gamma) = \Gamma'; \sigma_{\Gamma'}$, $sk_L(\Phi;\Delta)
  = \Delta'; \sigma_{\Delta'}$ and $sk_R(\Phi; F)=
  K; \sigma_K$.
  Let $\tau = \sigma_{\Gamma'}, \sigma_{\Delta'}, \sigma_K$.
Let $\Phi' = (FV(\Gamma') \cup FV(\Delta') \cup FV(\Phi_F)) \setminus \Phi$.
Assume that $\sigma$ does not contain any bound variables of $\Gamma$,
$\Gamma'$, $\Delta$, $\Delta'$, $F$ or $K$. Moreover, 
assume whenever $\Phi$ contains a variable $a_L$ or $a_R$, then the
corresponding variable $a_R$ or $a_L$ respectively does not occur in
$\Phi$.
Then there exists a substitution $\sigma' \colon \Phi, \Phi'
\rightarrow \Phi'$ such that
\[
\nega{\Gamma'}; \Delta' \skderiv K;\sigma, \tau,\sigma'\; .
\]
\end{theorem}
\begin{proof}
  Induction over the derivation of $\Gamma\sigma_{\uparrow \weak{\Phi}};
  \Delta\sigma_{\uparrow \weak{\Phi}} \vdash F\sigma_{\uparrow \weak{\Phi}}$.  The axiom case follows from
  the definition of admissibility, $\tensor$R follows from 
  Lemma~\ref{lemma:tensor}, and $\lolli$L from
  Lemma~\ref{lemma:lolli}.
  Now we consider the case of $\forall$L. By definition, $\skl{\Phi}
    {\forall x.F} = \skl{(x, \Phi)}{F}$. Moreover, $t$ contains only
    variables in $\Phi$. Hence we can apply the induction hypothesis
    with replacing $\Phi$ by $\Phi, x$.
    The next case is $\forall$R. Consider any formula $\forall
    u.F$. Skolemisation introduces another Eigen-variable $u$.
    Hence we can apply the induction hypothesis with
    replacing $\Phi$ by $\Phi, u$.
The case for $copy$ is a direct consequence of
  Lemma~\ref{lemma:copy}. All other cases are immediate. \hfill $\Box$ 
\end{proof}

\subsection{Completeness} \label{sec:comp}
We now prove the completeness direction of skolemisation, which means
that we can turn a proof in SLJF directly into a proof in LJF, by
inserting at appropriate places quantifier rules, as captured by the
constraints.
We introduce an order relation to capture constraints on the order of
rules in the proof.
\begin{definition}
For any substitution $\sigma$,  define an order $<$ by $x < u$ or $x < a$ if $a$ or $u$ occur in
 $x\sigma$, and $u < x$ or $u < a$ if the variable $x$ or $a$ occurs
 in $u(z_1, \ldots, z_n)$.
\end{definition}

\begin{lemma} [Strengthening]
  \label{lemma:Strengthening}
  \begin{enumr}
\item Assume $\Gamma, (a';\Phi; \sigma'):K; \Delta_1 \skderiv
K'; \sigma$ and there exists a free variable $x$ in $K$ such that $a_R$ occurs
in $x\sigma$. Moreover assume that $a_L$ occurs in every axiom of
$K'$. Then also $\Gamma; \Delta_1 \skderiv
K' \sigma$.
\item
  Assume $\Gamma, (a';\Phi; \sigma'):K; \Delta_2 \skderiv
K'; \sigma$ and there exists a free variable $x$
in $K$ such that $a_L$ occurs in $x\sigma$. Then also $\Gamma; \Delta_2 \skderiv
K'; \sigma$.
\end{enumr}
\end{lemma}
\begin{proof}
(i) If the copy-rule for $K$ is applied during the derivation, the linear context
 contains the free variable $x$ such that $a_R$ occurs in
 $x\sigma$. As $a_L$ occurs in all atoms of $K'$, the variable $x$
 must not occur in any of the linear formulae in the  axioms in
 the derivation of $\Gamma, (a';\Phi; \sigma'):K; \Delta_1 \skderiv
K'; \sigma$ because of the admissibility
 condition. Hence no subformula of $K$ can occur in the linear
 formulae in the axioms in this derivation either. 
Hence there is also a
 derivation of $\Gamma; \Delta_1 \skderiv
 K'; \sigma$, which does not involve $K$.
(ii)
  A similar argument applies. \hfill $\Box$
\end{proof}
\begin{lemma}
  \label{lemma:tensorCompl}
  Assume $\Gamma; \Delta_1, \Delta_2 \skderiv \up(K_1 \tensor K_2);\sigma$ .  Furthermore assume that each formula $K$ in
  $\Delta_1$ and $\Delta_2$ is either a
  formula $\down K'$, or there exists a free existential variable $x$ in $K$
  such that $a_L$ or $a_R$ occurs in $x\sigma$, where $a_L$ and $a_R$
  are the special variables introduced by the skolemisation of $K_1
  \tensor K_2$.
  Moreover assume that the first
  focusing rule applied is the focus R-rule.
Then $\Gamma; \Delta_1 \skderiv
K_1; \sigma$ and $\Gamma; \Delta_2 \skderiv
K_2; \sigma$.
\end{lemma}
\begin{proof}
  We use an induction over the structure of $\Delta_1$ and $\Delta_2$.
  Firstly, consider the case  $\Gamma; K_1' \tensor K_2', \Delta_1, \Delta_2 \skderiv \up(K_1 \tensor K_2);\sigma$.
  We have a derivation
  \[
    \infer{\Gamma; K_1' \tensor K_2', \Delta_1, \Delta_2 \skderiv \up(K_1 \tensor K_2);\sigma}
       {\Gamma; K_1', K_2', \Delta_1, \Delta_2 \skderiv\up(K_1 \tensor K_2)'\sigma}
  \]
By induction hypothesis we have $\Gamma; \Delta_1'; \skderiv
K_1; \sigma$ and $\Gamma; \Delta_2' \skderiv
K_2; \sigma$. Assume $a_L$ occurs in $x\sigma$.
Because $\sigma$ is
admissible for $\Gamma; \Delta_2'$, $K_1'$ and $K_2'$ must be part of
$\Delta_1'$. Hence $\Delta_1' = K_1', K_2', \Delta_1$ and $\Delta_2' =
\Delta_2$. An application of the $\tensor L$-rule now produces $
 \Gamma; K_1' \tensor K_2', \Delta_1; \skderiv K_1; \sigma$.

Next we consider the case $ \Gamma;
!_{(a'\Phi; \sigma')}K, \Delta_1, \Delta_2 \skderiv
\up(K_1 \tensor K_2); \sigma$. Assume without loss of
generality $a_R$ occurs in $x\sigma$. We have a derivation
  \[
    \infer{\Gamma; !_{(a';\Phi; \sigma')}K, \Delta_1, \Delta_2 \skderiv
      \up(K_1 \tensor K_2); \sigma}
       {\Gamma, (a';\Phi;\sigma'):K; \Delta_1, \Delta_2 \skderiv\up(K_1
         \tensor K_2); \sigma}
  \]
By induction hypothesis we have $\Gamma, (a';\Phi; \sigma'):K; \Delta_1; \skderiv 
K_1; \sigma$ and $\Gamma, (a';\Phi; \sigma'):K; \Delta_2 \skderiv
K_2; \sigma$.
An application of the $!L$-rule yields $\Gamma; !_{(a';\Phi; \sigma')}K,
\Delta_1; \skderiv K_1; \sigma$
and Lemma~\ref{lemma:Strengthening} yields $\Gamma; \Delta_2
\skderiv K_2; \sigma$. \hfill $\Box$
\end{proof}

\begin{lemma}
  \label{lemma:lolliCompl}
  Assume $\Gamma; \Delta_1, \Delta_2, \down(K_1 \lolli K_2)\skderiv K; \sigma$ .  Furthermore assume that each formula
  $K'$ in $\Delta_1$, $\Delta_2$ and $K$ is either a
  formula $\down K''$, or there exists a free existential variable $x$ in $K'$
  such that $a_L$ or $a_R$ occurs in $x\sigma$. Moreover assume that the first
  focusing rule applied is the focus L-rule for $K_1 \lolli K_2$.
Then $\Gamma; \Delta_1 \skderiv
K_1; \sigma$ and $\Gamma; \Delta_2, K_2 \skderiv
K; \sigma$.
\end{lemma}
\begin{proof}
  Similar to the proof of Lemma~\ref{lemma:tensorCompl}. \hfill $\Box$
\end{proof}

\begin{lemma}
  \label{lemma:plingR}
  Assume $\Gamma; \Delta \vdash \downarrow !_{(a, \phi; \sigma')} K; \sigma$ and
  the first occurrence of the focus-rule is the focus R-rule followed
  by $!R$ with $\Gamma'$ containing the side formulae. Let $x$ be a
  free variable $x$ of $\Gamma$, $\Delta$ or $!_{(a, \phi; \sigma')} K$.
  \begin{enumr}
    \item If the variable $u$ occurs in $x\sigma$, then $u$ is a free
      variable of $\Gamma'$ or $!_{(a, \phi; \sigma')} K$.
      \item The variable $a$ does not occur in $x\sigma$.
  \end{enumr}
\end{lemma}
\begin{proof}
     (i) By induction over the number of steps before application of
  the focus R-rule. Assume that the first rule applied is the
  focus R-rule.
There are several cases. Firstly, assume $u$ occurs bound in $\Gamma$. We consider here only the case that $u$ occurs in
      $(a_1, \Phi_1, \sigma_1):N_1$, which is part of $\Gamma$; all
      other cases are similar. By assumption we have $u < a_1$ and $x
      < u$. The $!R$-rule implies $a_1 < a$. If $x$ occurs freely in
      $\Gamma$, we also have $a < x $ via the $!R$-rule, which is a
      contradiction. If $x$ occurs freely in $K$, then we also have
      $a_1 < x$ via the $!R$-rule, which is a contradiction.
      Secondly, assume $u$ occurs bound in $K$. Hence
      $x$ cannot be a free variable of $K$. In this
      case we have $u < a$ and $x < u$ by assumption, together with $a
      <x$  by the $!R$-rule, which is a contradiction.
      The step case is true because there are fewer free variables in
      the conclusion of a rule than in the premises.
    
    (ii) Assume $x <  a$. Then there must exist a $u$ such that $x <u$ and
     $u < a$. The latter implies $u$ is a bound variable in $K$, which
     is a contradiction to $(i)$.
  
\end{proof}

\begin{theorem}[Completeness]
Let $\Phi$ be a set 
  of Eigen-, special, and existential  variables which contains all the free
  variables of $\Gamma$, $\Delta$ and $F$. 
 Let $\sigma\colon \Phi \rightarrow \Phi$ be a
  substitution. 
Let $sk_L(\Phi; \Gamma) = (\Gamma'; \sigma_{\Gamma'})$, $sk_L(\Phi; \Delta)
  = (\Delta';\sigma_{\Delta'})$ and $sk_R(\Phi; F)=
  (K; \sigma_K)$.  
Let $\Phi' = (FV(\Gamma') \cup FV(\Delta') \cup FV(K)) \setminus \Phi$
and $\tau = \sigma_{\Gamma'}, \sigma_{\Delta'}, \sigma_K$.
Let $\sigma'\colon \Phi,\Phi'\rightarrow \Phi'$ be a
substitution.
\begin{enumr}
  \item
If  
$ \nega{\Gamma'}; \Delta' \skderiv K;
\sigma, \tau, \sigma'
$
then $\Gamma\sigma_{\uparrow \weak{\Phi}}; \Delta\sigma_{\uparrow
  \weak{\Phi}} \vdash F\sigma_{\uparrow
  \weak{\Phi}}$ in focused intuitionistic linear logic.
\item
If
$\Delta' = \Delta'', \downarrow K'$ and
$
\nega{\Gamma'}; \Delta'', [K'] \skderiv K;
\sigma, \tau, \sigma'
$
then $\Gamma\sigma_{\uparrow \weak{\Phi}}; \Delta\sigma_{\uparrow
  \weak{\Phi}} \vdash F\sigma_{\uparrow
  \weak{\Phi}}$ in focused intuitionistic linear logic.
\item
  If 
$\nega{\Gamma'}; \Delta' \skderiv [K];
\sigma, \tau, \sigma'$
then  $\Gamma\sigma_{\uparrow \weak{\Phi}}; \Delta\sigma_{\uparrow
  \weak{\Phi}} \vdash F\sigma_{\uparrow
  \weak{\Phi}}$ in focused intuitionistic linear logic.
\end{enumr}
\end{theorem}

\begin{proof}
We use firstly an induction over the derivation of 
$
\nega{\Gamma'}; \Delta' \skderiv K;
\sigma, \tau, \sigma'$ and secondly an induction over the structure of
$\Delta, F$.
Let $\Delta = F_1, \ldots, F_n$ and $\Delta' = K_1, \ldots, K_n$.
  Let $V = \{x_1,\ldots,$ $ x_k, u_1, \ldots, u_m\}$ be the set of outermost bound
variables of $\Delta', K$ (including names).
There are several cases. Firstly, if there exists a $i$ such that $1 \leq i \leq n$ and
$F_i$ is a tensor product or a formula $!N$, or $F$ is
a linear implication, we apply the corresponding inference rule and
then the induction hypothesis.

Secondly, assume there exists an Eigen-variable $u \in V$. 
Assume $F = \forall u.F'$.
Hence by 
  induction hypothesis  we have $\Gamma\sigma_{\uparrow \weak{\Phi}};
  \Delta\sigma_{\uparrow   \weak{\Phi}} \vdash F'\sigma_{\uparrow
  \weak{\Phi}}$. By assumption, $u$ does not occur in $x\sigma$ for
any variable $x$ in the co-domain of $\sigma$.
Now the $\forall R$-rule yields the claim.
 Now assume $F = \exists u.F'$. This case is similar to $\forall R$.

Thirdly, assume there exists an existential variable in $V$.  Let $x$ be an existential variable which is maximal in $V$.
Assume $F =  \exists x.F'$.
We show that every Eigen-variable $u$ of $x\sigma'$ is a free 
variable of $\Delta, F$. By definition, we have $x
<u$. Assume $u$ is a bound variable in $\Delta, F$. If $u$
is a bound variable of $F$, we would have $u < x$, which is a
contradiction. Hence $u$ is a bound variable of $\Delta$.
Because $u$ is not an outermost bound variable, there exists a
bound existential variable  $y$ such that $u < y$. Hence $x$ is not a maximal
bound variable.
By induction hypothesis we have
$\Gamma\sigma_{\uparrow \weak{\Phi}}; \Delta\sigma_{\uparrow
  \weak{\Phi}} \vdash F'\sigma_{\uparrow
  \weak{\Phi}}$, and now we apply the $\exists R$-rule. 
  Now assume $F_1 =  \forall x.F_1'$. Similar to the $\exists R$-case.

Next, assume there are no maximal first-order variables in $V$. By definition,
the special variables corresponding to the last rule applied to
the skolemised version where the
principal formula is asynchronous are now the only maximal elements in $V$. 
$\tensor R$ and $\lolli L$ are direct consequences of Lemma~\ref{lemma:tensorCompl} and
Lemma~\ref{lemma:lolliCompl}, respectively. For $!R$, let $x$ be any
outermost bound variable in $\Gamma$, $\Delta$ or $K$ which is not maximal
in $V$. Because $x \not < a$, there exists a variable $y$ or $u$ in
$V$ such that $x < y$ or $x < u$, which is a contradiction.
Hence we can use the $!R$-rule of the skolemised calculus and the induction hypothesis. Finally, the axiom rule in the skolemised calculus implies $n=1$, and hence  $\Gamma\sigma_{\uparrow \weak{\Phi}};F_1\sigma_{\uparrow \weak{\Phi}} \vdash F\sigma_{\uparrow \weak{\Phi}}$.
\hfill $\Box$
\end{proof}

\section{Conclusion}
\label{sec:conc}

In this paper, we revisit the technique of skolemisation and adopt it for proof search in first-order focused and polarised intuitionistic linear logic (LJF).  The central idea is to encode quantifier dependencies by constraints, and the global partial order in which quantifier rules have to be applied by a substitution.   We propose a domain specific logic called SLJF, which avoids back-tracking during proof search when variable instantiations are derived by unification.   

\emph{Related work:} Shankar~\cite{Shankar91proofsearch} first propose an adaptation of skolemisation to LJ.
Our paper can be seen as a generalisation of this work to focused and polarised linear logic.
Reis and Paleo~\cite{Reis17ifcolog} propose a technique called epsilonisation to characterise the permutability of rules in LJ.
Their approach is elegant but impractical, because it trades an exponential growth in the search space with an exponential growth in the size of the proof terms.
McLaughlin and Pfenning~\cite{conf/lpar/McLaughlinP08} propose an effective proof search technique based on the inverse method for focused and polarised intuitionistic logic.
To our knowledge, the resulting theorem prover Imogen~\cite{CADE-2009-McLaughlinP} would benefit from the presentation of skolemisation in our paper, since it requires backtracking to resolve the first-order non-determinism during proof search.

\emph{Applications:} There are ample of applications for skolemisation.
To our knowledge, proof search algorithms for intuitionistic or substructural logic are good at removing non-determinism from the propositional level, but don't solve the problem at the first-order level.
Skolemisation can also be applied to improve intuitionistic theorem provers further, such as Imogen.
With the results in this paper we believe that we are able to achieve such results without much of a performance penalty.

\bibliographystyle{plain}
\bibliography{skolem}
\end{document}